\documentclass[aip,pop,preprint,a4paper]{revtex4-1}
\usepackage{cancel}
\usepackage{amsmath}
\usepackage{amssymb}
\usepackage{amsthm}
\usepackage{float}
\usepackage{graphicx}
\usepackage{dcolumn}
\usepackage{bm}
\usepackage{upgreek}
\usepackage{esint}
\usepackage{subcaption}
\usepackage{xcolor}

\theoremstyle{definition}

\theoremstyle{remark}

\theoremstyle{plain}

\theoremstyle{plain}

\theoremstyle{plain}

\theoremstyle{plain}
\newtheorem{proposition}{Proposition}

\begin{document}


\title{Exact expressions for nonperturbative guiding center theory in symmetric fields}



\author{I. Hollas}
\affiliation{Department of Physics, The University of Texas at Austin, Austin, TX 78712, USA}
\author{R. Agarwal}
\affiliation{Department of Physics, The University of Texas at Austin, Austin, TX 78712, USA}
\author{J. W. Burby}
\affiliation{Department of Physics and Institute for Fusion Studies, The University of Texas at Austin, Austin, TX 78712, USA}
\author{A. J. Brizard}
\affiliation{Department of Physics, Saint Michael's College, Colchester, VT 05439, USA}


\date{\today}

\begin{abstract}
We apply a recently-developed nonperturbative guiding center formalism to charged particle dynamics in fields with two-parameter continuous symmetry groups. This entails finding exact constants of motion, valid in the nonperturbative regime, that agree with Kruskal's adiabatic invariant series to all orders in the perturbative regime, when the field scale length is large compared with a typical gyroradius. We demonstrate that the nonperturbative guiding center model makes exact predictions in these cases, even though it eliminates the cyclotron timescale, thereby establishing a theoretical baseline for performance of the nonperturbative formalism.
\end{abstract}


\maketitle 

\section{Introduction}
Perturbative guiding center theory, as developed by Kruskal \cite{kruskal58,Kruskal_1962}, Littlejohn \cite{Littlejohn_1981,Littlejohn_1983,Littlejohn_1984}, and various others \cite{Cary_2009}, provides a theoretical backbone for much of magnetized plasma modeling. However, certain classes of high-energy plasma particles \cite{Carbajal_2024,Rodrigues_2024,Assuncao_2023,Carbajal_2020,Ogawa_2016,Pefferle_2015,Cecconello_2018} break the usual guiding center perturbation expansions. For such particles, a recently proposed non-perturbative variant of the guiding center model \cite{j_w_burby_nonperturbative_2025} shows promise as a replacement for traditional guiding center theory that achieves accuracy comparable to the full-orbit model without resolving the short cyclotron timescale.

The non-perturbative guiding center model differs qualitatively from the traditional model because it requires first finding the non-perturbative adiabatic invariant $\mathcal{J}$. Here, ``non-perturbative adiabatic invariant" refers to an approximate constant of motion that agrees with Kruskal's adiabatic invariant series \cite{Kruskal_1962,burby_general_2020} to all orders in the perturbative regime $\epsilon\ll 1$, while remaining nearly-conserved in the nonperturbative regime $\epsilon\sim 1$. Here $\epsilon = \rho/L$ denotes the ratio of a characteristic gyroradius to a characteristic field scale length. Once the invariant is known the non-perturbative guiding center equations are explicit and readily computable. For general magnetic geometries Ref.\,\onlinecite{j_w_burby_nonperturbative_2025} advocates a data-driven approach to finding the non-perturbative adiabatic invariant. However, this approach assumes existence of $\mathcal{J}$ based on empirical numerical evidence. Rigorous theoretical analysis establishing existence of $\mathcal{J}$ from first principles would place the non-perturbative guiding center model on firmer theoretical footing.

Establishing existence of an exact non-perturbative adiabatic invariant in general magnetic geometries may lie beyond the reach of any analytical method. Indeed, a non-perturbative exact invariant cannot exist in regions of phase space where dynamics is sufficiently chaotic. On the other hand, it is obvious that a non-perturbative $\mathcal{J}$ exists for particles moving in a straight magnetic field $\bm{B} = B_0\,e_z$. This suggests constructing an exact non-perturbative $\mathcal{J}$ may be possible in sufficiently symmetric field configurations. In fact, Qin and Davidson found an exact invariant \cite{qin_exact_2006} asymptotic to the leading-order magnetic moment in fields of the form $\bm{B}=B(t)\,e_z$. This invariant does not clearly comprise an example of the true non-perturbative $\mathcal{J}$ because Qin and Davidson only demonstrated agreement with Kruskal's series at leading order in the perturbative regime. Nevertheless, their result motivates a search for exact invariants that furnish a non-perturbative $\mathcal{J}$ in special field configurations.

In this Article we provide a complete justification for the non-perturbative guiding center model in a pair of symmetric field configurations: a slab $\bm{B} = (1+y)e_z$ and a screw pinch $\bm{B} = \nabla\psi\times\nabla\theta - \iota(\psi)\nabla\psi\times\nabla z$. In each case, we first identify a formula for an exact invariant $\mathcal{J}$ and prove it agrees with Kruskal's adiabatic invariant series to all orders in $\epsilon$ when $\epsilon\ll 1$. Crucially, the $\mathcal{J}$ we construct is an exact constant of motion, even for non-perturbative $\epsilon$. We establish all-orders agreement with Kruskal's series by exploiting the well-known Liouville-Arnold theorem \cite{Arnold_1989,zung_conceptual_2018} and Kruskal's foundational observations on uniqueness properties of the adiabatic invariant series. Then we construct the non-perturbative guiding center model associated with $\mathcal{J}$ and show numerically that it agrees with the full-orbit model.

In each of the examples considered below we consider the motion of a charged particle with mass $m$ and charge $q$. The constants $B_0$, $k_0$, and $\rho_0$ denote a characteristic magnetic field strength, magnetic field wavenumber, and gyroradius, respectively. The characteristic cyclotron (angular) frequency is $\omega_0 = |q|\,B_0/m$, the guiding center ordering parameter is $\epsilon = k_0\,\rho_0$, and the sign of the charge is $\sigma$. We measure distance, velocity, time, and magnetic field strength in units of $k_0^{-1}$, $\omega_0\,\rho_0$, $\omega_0^{-1}$, and $B_0$.


\section{Uniform field}
We consider the simple case of a uniform magnetic field first because it motivates our identification of action variables in the slab and screw pinch. The magnetic field is given by $\bm{B} = \bm{e}_z$. For simplicity, we assume that the $x$--direction is periodic with periodicity $2\pi$.  The equations of motion  are given by 
\begin{align}
\dot{v}_x & = \phantom{-}\sigma\,v_y\label{uniform_vxdot}\\
\dot{v}_y & = -\sigma\,v_x\label{uniform_vydot}\\
\dot{x} & = \epsilon\,v_x\label{uniform_xdot}\\
\dot{y} & = \epsilon\,v_y.\label{uniform_ydot}
\end{align}
Note that we ignore dynamics in the $z$-direction, which trivially decouple from the $(x,y)$-dynamics. 

The equations of motion \eqref{uniform_vxdot}-\eqref{uniform_ydot} comprise a Hamiltonian system on the symplectic manifold $M = \mathbb{T}\times \mathbb{R}\times\mathbb{R}^2\ni (x,y,v_x,v_y)$, where $\mathbb{T}$ denotes the $2\pi$-periodic circle $\mathbb{T} = \mathbb{R}/2\pi\mathbb{Z}$. The symplectic $2$-form is $\omega = - d\vartheta$, where the Liouville $1$-form is
\begin{align*}
\vartheta = \epsilon\,(v_x\,dx + v_y\,dy) -\sigma\,y\,dx,
\end{align*}
and the Hamiltonian is $H = \epsilon^2\frac{1}{2}(v_x^2 + v_y^2)$. The Noether conserved quantity associated with $x$-translation invariance is therefore
\begin{align*}
p_x = \iota_{\partial_x}\vartheta = \epsilon\,v_x - \sigma\,y.
\end{align*}
The pair of constants of motion $H$ and $p_x$ commute under Poisson bracket and are functionally-independent. It follows that the dynamics is integrable in the sense of Liouville. Liouville integrability implies that each connected component of any compact regular level set of $(H,p_x)$ is an invariant $2$-torus, leading to a view of phase space foliated by invariant $2$-tori.

In order to proceed further it is useful to introduce alternate constants of motion, $(r,Y)$, and change coordinates on phase space from $(x,y,v_x,v_y)$ to $(x,y,r,\zeta)$. The conserved quantities are $r = \epsilon^{-1}\sqrt{2H}$ and $Y = -p_x/\sigma$. The old phase space coordinates relate to the new ones according to $v_x = r\cos\zeta$, $v_y = -r\sin\zeta$ 

Assuming $\overline{r} > 0$, the level set $\Gamma(\overline{r},\overline{Y})$ defined by $(r,Y) = (\overline{r},\overline{Y})$ is diffeomorphic to the $2$-torus $\mathbb{T}^2$. We introduce a parameterization of this $2$-torus, $\mathbb{T}^2\rightarrow \Gamma(\overline{r},\overline{Y}):(x,\zeta)\mapsto \Gamma(x,\zeta | \overline{r},\overline{Y})$, according to
\begin{align*}
\Gamma(x,\zeta | \overline{r},\overline{Y}) = \begin{pmatrix}x \\ y^*(x,\zeta | \overline{r},\overline{ Y}) \\  r^*(x,\zeta | \overline{r},\overline{ Y}) \\ \zeta\end{pmatrix},
\end{align*}
where the component functions are 
\begin{align*}
r^*(x,\zeta|\overline{r},\overline{Y}) & = \overline{r}\\
y^*(x,\zeta | \overline{r},\overline{Y}) & = \overline{Y} + \sigma\,\epsilon\,\overline{r}\, \cos\zeta.
\end{align*}
A homology basis for $\Gamma(r,Y)$ is given by the pair of closed curves $\gamma_i$ defined according to
\begin{align*}
\gamma_1(\zeta| r,Y) &= \Gamma(0,\sigma\,\zeta\mid r,Y)\\
\gamma_2(x | r,Y) & =\Gamma(x,0\mid r,Y).
\end{align*}
(Note that the orientation of $\gamma_1$ agrees with that of $\epsilon=0$ cyclotron rotation.) Integrating the Liouville $1$-form along these curves then defines the action variables $J_i(r,Y) = (2\pi)^{-1}\oint_{\gamma_i}\vartheta$, given explicitly by
\begin{align}
J_1(r,Y) & = \frac{\epsilon^2\,r^2}{2}\label{slab_J}\\
J_2(r,Y) & = -\sigma\,Y.
\end{align}
As is well-known, the action variable $\mathcal{J} = J_1$ is the exact magnetic moment invariant. 

\section{Slab geometry}
\subsection{Nonperturbative adiabatic invariant}
The magnetic field is given by $\bm{B} = (1+y)\bm{e}_z$. The region of interest is $y>-1$ 
, where the magnetic field is positive 
For simplicity, we assume that the $x$- and $z$-directions are each periodic with periodicity $2\pi$ 
The equations of motion  are given by 
\begin{align}
\dot{v}_x & = \phantom{-}\sigma\,v_y\,(1+y)\label{slab_vxdot}\\
\dot{v}_y & = -\sigma\,v_x\,(1+y)\\
\dot{x} & = \epsilon\,v_x\\
\dot{y} & = \epsilon\,v_y.\label{slab_ydot}
\end{align}
Note that we ignore dynamics in the $z$-direction, which trivially decouple from the $(x,y)$-dynamics. Brizard gives a complete derivation of the action-angle coordinates for this system in Ref.\,\onlinecite{brizard_actionangle_2022}.

The equations of motion \eqref{slab_vxdot}-\eqref{slab_ydot} comprise a Hamiltonian system on the symplectic manifold \cite{Abraham_2008} $M = \mathbb{T}\times \mathbb{R}\times\mathbb{R}^2\ni (x,y,v_x,v_y)$, where $\mathbb{T}$ denotes the $2\pi$-periodic circle $\mathbb{T} = \mathbb{R}/2\pi\mathbb{Z}$. The symplectic $2$-form is $\omega = - d\vartheta$, where The Liouville $1$-form is
\begin{align*}
\vartheta = \epsilon\,(v_x\,dx + v_y\,dy) -\sigma(y + y^2/2)dx,
\end{align*}
and the Hamiltonian is $H = \epsilon^2\frac{1}{2}(v_x^2 + v_y^2)$.
The Noether conserved quantity associated with $x$-translation invariance is therefore
\begin{align*}
p_x = \iota_{\partial_x}\vartheta = \epsilon\,v_x - \sigma\,(y + y^2/2).
\end{align*}
The pair of constants of motion $H$ and $p_x$ commute under Poisson bracket and are functionally-independent. It follows that the dynamics is integrable in the sense of Liouville. Liouville integrability implies that each connected component of any compact regular level set of $(H,p_x)$ is an invariant $2$-torus, leading to a view of phase space foliated by invariant $2$-tori.

In order to proceed further it is useful to introduce alternate constants of motion, $(r,Y)$, and change coordinates on phase space from $(x,y,v_x,v_y)$ to $(x,y,r,\zeta)$. The conserved quantities are $r = \epsilon^{-1}\sqrt{2H}$ and $Y = -p_x/\sigma$. The old phase space coordinates relate to the new ones according to $v_x = r\cos\zeta$, $v_y = -r\sin\zeta$.


Assuming $\overline{r} > 0$ and $\overline{Y} > 0$, the level set $\Gamma(\overline{r},\overline{Y})$ defined by $(r,Y) = (\overline{r},\overline{Y})$ is diffeomorphic to the $2$-torus $\mathbb{T}^2$. We introduce a parameterization of this $2$-torus, $\mathbb{T}^2\rightarrow \Gamma(\overline{r},\overline{Y}):(x,\zeta)\mapsto \Gamma(x,\zeta | \overline{r},\overline{Y})$, according to
\begin{align*}
\Gamma(x,\zeta | \overline{r},\overline{Y}) = \begin{pmatrix}x \\ y^*(x,\zeta | \overline{r},\overline{ Y}) \\  r^*(x,\zeta | \overline{r},\overline{ Y}) \\ \zeta\end{pmatrix},
\end{align*}
where the component functions are 
\begin{align*}
r^*(x,\zeta|\overline{r},\overline{Y}) & = \overline{r}\\
y^*(x,\zeta | \overline{r},\overline{Y}) & = -1+ (1 + 2\overline{Y})^{1/2}\left(1 + \frac{2\,\sigma\,\epsilon\,\overline{r}}{(1 + 2\overline{Y})}\cos\zeta\right)^{1/2}.
\end{align*}
A homology basis for $\Gamma(r,Y)$ is given by the pair of closed curves $\gamma_i$ defined according to
\begin{align*}
\gamma_1(\zeta| r,Y) &= \Gamma(0,\sigma\zeta\mid r,Y)\\
\gamma_2(x | r,Y) & =\Gamma(x,0\mid r,Y).
\end{align*}
(Note that the orientation of $\gamma_1$ agrees with that of $\epsilon=0$ cyclotron rotation.) Integrating the Liouville $1$-form along these curves then defines the action variables $J_i(r,Y) = (2\pi)^{-1}\oint_{\gamma_i}\vartheta$, given explicitly by
\begin{align}
J_1(r,Y) & = \epsilon\,\sigma\,r\,(1+2Y)^{1/2}\frac{1}{2\pi}\int_0^{2\pi}\left(1 + \frac{2\,\sigma\,\epsilon\,r}{(1 + 2Y)}\cos\zeta\right)^{1/2}\,\cos\zeta\,d\zeta\label{slab_J}\\
J_2(r,Y) & = -\sigma\,Y.
\end{align}


We claim that the first action variable $J_1$ is the non-perturbative adiabatic invariant for this system. We argue precisely as follows.

\begin{proposition}\label{slab_adinv_prop}
The constant of motion $\mathcal{J} = J_1$, given in Eq.\,\eqref{slab_J}, is a non-perturbative adiabatic invariant for Eqs.\,\eqref{slab_vxdot}-\eqref{slab_ydot}. In particular, the series expansion of $J_1$ in powers of $\epsilon$ agrees with Kruskal's adiabatic invariant series to all orders in $\epsilon$.
\end{proposition}
\begin{proof}
By the Liouville-Arnold theorem, in a neighborhood of $\Gamma({r,Y})$ there are action-angle variables $(\theta_1,\theta_2,J_1,J_2)$, in which the Liouville $1$-form is $\vartheta = J_1\,d\theta_1 + J_2\,d\theta_2$, modulo closed $1$-forms. It follows that the Hamiltonian vector field associated with $J_1$, $X_{J_1}$, is the infinitesimal generator for a circle-action on phase space that leaves Lorentz force dynamics invariant. 

First we will show that as $\epsilon$ tends to zero $X_{J_1}$ limits to so-called limiting roto-rate $\mathcal{R}_0 = \sigma(v_y\,\partial_{v_x} - v_x\,\partial_{v_y})$. The differential of the action $J_1$ is given by $dJ_1 = \partial_rJ_1\,dr + \partial_{Y}J_1\,d Y$. The differentials $dr$ and $d Y$ may be expressed in terms of $dH$ and $dp_x$ as $dr = \epsilon^{-2}(1/r)\,dH$ and $dY = -\sigma\,dp_x$. The Hamiltonian vector field $X_{J_1}$ is therefore given by
\begin{align*}
X_{J_1}=\partial_rJ_1\,\left(\frac{1}{\epsilon^2\,r}\right)\,X_H - \partial_YJ_1\,\sigma\,\partial_x.
\end{align*}
The function $J_1(r,Y)$ can be expanded in powers of $\epsilon$ as
\begin{align}
J_1(r,Y) = \frac{\epsilon^2\,r^2}{2(1+2Y)^{1/2}} + \frac{3\epsilon^4\,r^4}{16(1+2Y)^{5/2}} + O(\epsilon^5),\label{J1_expansion}
\end{align}
which implies the derivatives $\partial_rJ_1$, $\partial_YJ_1$ have the $\epsilon$-expansions
\begin{align*}
\partial_rJ_1(r,Y)& = \frac{\epsilon^2\,r}{(1+y)} + \frac{\sigma\epsilon^3\,r^2\,\cos\zeta}{(1+y)^3} + O(\epsilon^4)\\
\partial_YJ_1(r,Y) & = -\frac{\epsilon^2\,r^2}{2(1+y)^3}  - \frac{3\sigma\,\epsilon^3\,r^3\,\cos\zeta}{2(1+y)^5} + O(\epsilon^4).
\end{align*}
(Here it is crucial that the derivatives $\partial_r,\partial_Y$ are computed before substituting the expression for $Y=Y(y,r,\zeta)$.) The limiting value of $X_{J_1}$ as $\epsilon\rightarrow 0$ is therefore
\begin{align}
\lim_{\epsilon\rightarrow 0}X_{J_1} & = \lim_{\epsilon\rightarrow 0}\bigg(\partial_rJ_1\,\left(\frac{1}{\epsilon^2\,r}\right)\,X_H\bigg) - \lim_{\epsilon\rightarrow 0}\bigg(\partial_YJ_1\,\sigma\,\partial_x\bigg)\nonumber\\
& = \lim_{\epsilon\rightarrow 0}\bigg(\frac{\epsilon^2\,r}{(1+y)}\,\left(\frac{1}{\epsilon^2\,r}\right)\,X_H\bigg) + \lim_{\epsilon\rightarrow 0}\bigg(-\frac{\epsilon^2\,r^2}{2(1+y)^3}\,\sigma\,\partial_x\bigg)\nonumber\\
& = \lim_{\epsilon\rightarrow 0}\bigg(\frac{1}{(1+y)}\,X_H\bigg) \nonumber\\
& = \sigma(v_y\partial_{v_x} - v_x\,\partial_{v_y}) = \mathcal{R}_0,\label{slab_limit_cond}
\end{align}
as claimed.

To complete the proof we now make use of a uniqueness result originally established by Kruskal \cite{Kruskal_1962}. Recall we want to show that $J_1$ agrees with Kruskal's adiabatic invariant series $\mu$ to all orders in $\epsilon$ when $\epsilon \ll 1$. Observe that it is enough to show that $X_{J_1}$ agrees with Kruskal's roto-rate vector field \cite{burby_general_2020} $\mathcal{R}$ to all orders in $\epsilon$. For if this were the case then $dJ_1 = \iota_{X_{J_1}}\omega = \iota_{\mathcal{R}}\omega = d\mu$, which implies $J_1$ and $\mu$ differ by an unimportant constant. (Here we use the fact \cite{burby_general_2020}
 that the roto rate is a Hamiltonian vector field with Hamiltonian $\mu$.) Kruskal showed that the roto-rate $\mathcal{R}$ is uniquely determined as a formal power series in $\epsilon$ by three conditions: (1) $[\mathcal{R},X_{H}] = 0$, (2) every integral curve for $\mathcal{R}$ is periodic with period $2\pi$, and (3) $\lim_{\epsilon\rightarrow 0}\mathcal{R} = \mathcal{R}_0$. We claim that $X_{J_1}$ satisfies each of these conditions. (1) follows from $[X_{J_1},X_{H}] = -X_{\{J_1,H\}} = 0$. (2) follows from the fact that $X_{J_1}$ is the infinitesimal generator for a circle action. (3) was established in Eq.\,\eqref{slab_limit_cond}. It follows that $X_{J_1}$ must agree with Kruskal's roto-rate $\mathcal{R}$, as claimed.

\end{proof}
 We note that although Kruskal's theory only gives formal power series roto-rates and adiabatic invariants in general, in this problem we obtain exact non-perturbative analogues of Kruskal's series due to complete integrability of the dynamics.

\subsection{Direct comparison with Kruskal's series}
{We will now explicitly compare the constant of motion \eqref{slab_J} with Kruskal's adiabatic invariant series. We temporarily assume $\sigma = 1$. In previous work \cite{brizard_actionangle_2022}, the action integral (11) is expressed explicitly as
\begin{equation}
    J_{1}({\sf e},u) \;=\; \frac{4}{3\pi}\,\nu_{0}^{3}
    \sqrt{(1 + {\sf e})^{3}}\left[ (2 - m) {\sf E}(m) \;-\frac{}{} 2\,
    (1 - m)\;{\sf K}(m)\right],
    \label{eq:J1_EK}
\end{equation}
where $1 + 2\,Y = 2\,u \equiv 4\,\nu_{0}^{2}$, ${\sf e} \equiv \epsilon r/u = {\cal O}(\epsilon)$, and $m = 2{\sf e}/(1 + {\sf e})$. Here, the complete elliptic integrals ${\sf E}(m)$ and ${\sf K}(m)$ are defined according to the Abromowitz and Stegun\cite{Abramowitz_1964} definitions, i.e.,
\[ {\sf K}(m) \;\equiv\; \int_{0}^{\pi/2} \frac{d\varphi}{\sqrt{1 - m\,\sin^{2}\varphi}}. \]
When the action integral (\ref{eq:J1_EK}) is expanded up to fourth order in ${\sf e}$, we obtain
\begin{equation}
   J_{1}({\sf e},u) \;=\; \nu_{0}^{3}{\sf e}^{2} \left( 1 \;+\; 
   \frac{3}{32}\,{\sf e}^{2} \;+\; {\cal O}({\sf e}^{5}) \right),
   \label{eq:J1_final}
\end{equation}
which is identical to Eq.~(\ref{J1_expansion}).

We now show explicitly that the action integral (\ref{slab_J}) is exactly equal to the guiding-center magnetic moment $\mu$, which is expressed in a perturbation expansion as $\mu = \mu_{0} + \mu_{1} + \mu_{2} + \cdots$, where the term $\mu_{n} = {\cal O}({\sf e}^{n+2})$. Here, the lowest-order magnetic moment and its first-order correction are
\begin{eqnarray}
    \mu_{0} &=& \frac{\epsilon^{2}r^{2}}{2\,(1 + y)} \;=\; 
    \frac{2\,\nu_{0}^{4}{\sf e}^{2}}{(1 + y)}, \label{eq:mu_0} \\
    \mu_{1} &=& \mu_{0}\,\rho_{0}\cdot\nabla\ln B \;=\; \frac{4\,\nu_{0}^{6}{\sf e}^{3}}{(1 + y)^{3}}\;\cos(2\varphi), \label{eq:mu_1}
\end{eqnarray}
where $1 + y \equiv 2\,\nu_{0}\sqrt{(1 + {\sf e}) - 2\,{\sf e}\,\sin^{2}\varphi} > 0$ and the lowest-order gyroradius is
\begin{equation}
    \rho_{0} \;\equiv\; \frac{2\,\nu_{0}^{2}{\sf e}}{(1 + y)}\; \left(
    \bm{e}_x\;\cos s \;-\frac{}{} \bm{e}_y\;\sin s \right).
\end{equation} 
Here, the gyroangle $s$ is defined as 
$s \equiv 2\,\varphi - \pi/2$ for convenience of comparison with Ref.\,\onlinecite{brizard_actionangle_2022}, which used the symbol $\zeta$ in place of $s$. 

The expression for the second-order correction to the magnetic moment is derived by Lie-transform perturbation method in Tronko \& Brizard\cite{tronko_lagrangian_2015}, where
\begin{equation}
    \mu_{2} \;\equiv\; G_{2}^{\mu} \;+\; \frac{1}{2} \left[-\;\rho_{0}\cdot\nabla\mu_{1} \;+\; \mu_{1}\;\frac{\partial\mu_{1}}{\partial\mu_{0}} \;+\; \left(\frac{\partial\rho_{0}}{\partial s}\cdot\nabla\ln B\right)\;\frac{\partial\mu_{1}}{\partial s}\right],
    \label{eq:mu2_def}
\end{equation}
with the second-order magnetic-moment generator defined as
\begin{equation}
    G_{2}^{\mu} \;\equiv\; -\,\frac{1}{2}\,\left(\frac{\partial\rho_{0}}{\partial s}\cdot\nabla\ln B\right)\;\frac{\partial\mu_{1}}{\partial s} \;-\; \frac{H_{{\rm gc}2}}{B}.
\end{equation}
When the second-order magnetic moment (\ref{eq:mu2_def}) is derived for a straight magnetic field with constant gradient, we obtain $H_{{\rm gc}2} = -\,(3/4)\,\mu_{0}^{2}/B^{2}$ (identical to Burby, Squire, and Qin\cite{burby_automation_2013}), and
\begin{equation}
    \mu_{2} \;=\; \frac{3\,\mu_{0}^{2}}{4\,B^{3}}\;\left( 1 \;+\frac{}{} 4\,
    \cos^{2}(2\varphi)\right) \;=\; \frac{3}{32}\,\frac{\nu_{0}^{3}{\sf e}^{4}}{(1 + y)^{5}}\;\left( 1 \;+\frac{}{} 4\,
    \cos^{2}(2\varphi)\right).
    \label{eq:mu2_final}
\end{equation}

We now proceed with expansions of $(1 + y)^{-1}$ and $(1 + y)^{-3}$ in Eqs.~(\ref{eq:mu_0})-(\ref{eq:mu_1}) up to second order and first order in ${\sf e}$, respectively, which yields
\begin{eqnarray}
    \mu_{0} &=& \nu_{0}^{3}{\sf e}^{2} \left( 1 \;-\; \frac{1}{2}\,{\sf e}\,
    \cos(2\varphi) \;+\; \frac{3}{8}\,{\sf e}^{2}\,\cos^{2}(2\varphi) \right), \label{eq:mu0_4} \\
    \mu_{1} &=& \nu_{0}^{3}{\sf e}^{2} \left( \frac{1}{2}\,{\sf e}\,
    \cos(2\varphi) \;-\; \frac{3}{4}\,{\sf e}^{2}\,\cos^{2}(2\varphi) \right).
    \label{eq:mu1_4}
\end{eqnarray}
In Eq.~(\ref{eq:mu2_final}), there is no need to expand $(1 + y)^{-5}$ since we are already at the highest order $({\sf e}^{4})$ considered, so that
\begin{equation} 
\mu_{2} \;=\; \nu_{0}^{3}{\sf e}^{2}\;\left( \frac{3}{32}\,{\sf e}^{2} 
\;+\frac{}{} \frac{3}{8}\,{\sf e}^{2}\,\cos^{2}(2\varphi)\right).
\label{eq:mu2_4}
\end{equation}
If we now add Eqs.~(\ref{eq:mu0_4})-(\ref{eq:mu2_4}), we readily find
\begin{equation}
    \mu \;=\; \mu_{0} + \mu_{1} + \mu_{2} \;=\; \nu_{0}^{3}{\sf e}^{2}\;\left( 1 \;+\; \frac{3}{32}\,{\sf e}^{2}\right),
\end{equation}
which is exactly equal to the expansion of the action integral (\ref{eq:J1_final}) up to fourth order in ${\sf e}$.
}

\section{Screw pinch}
\subsection{Nonperturbative adiabatic invariant}
The magnetic field is given by $\bm{B} = \nabla\psi\times\nabla\theta - \iota(\psi)\nabla\psi\times\nabla z$. Here $(r,\theta,z)$ denote standard cylindrical coordinates, where $\theta$ is the azimuthal angle. Note that the symbol $r$ was previously used to denote perpendicular velocity; there is no conflict because we do not refer to perpendicular velocity in this Section. We will treat $z$ as a periodic variable with period $2\pi$. We will also assume that $\psi = \psi(r)$ depends on radius only, and that the rotational transform $\iota(\psi) = \psi_P^\prime(\psi)$ is given as the derivative of a poloidal flux function $\psi_P(\psi)$.  The equations of motion are given by
\begin{align}
\dot{p_r}&=\sigma\psi'(r)(r^{-2}p_\theta-\iota(\psi(r))p_z)+\epsilon r^{-3}p_\theta^2\label{sp_prdot}\\
\dot{p_\theta}&=-\sigma p_r\psi'(r)\\
\dot{p_z}&=\sigma p_r\iota(\psi(r))\psi'(r)\\
\dot{r}&=\epsilon p_r\\
\dot{\theta}&=\epsilon r^{-2}p_\theta\\
\dot{z}&=\epsilon p_z,\label{sp_zdot}
\end{align}
where $(p_r,p_\theta,p_z)$ denote the covariant components of kinetic momentum associated with the cylindrical coordinate system, $\bm{p}=p_r\nabla r + p_\theta\nabla\theta + p_z\nabla z$. Note in particular that $(p_r,p_\theta,p_z)$ differ from the canonical momenta $(P_r,P_\theta,P_z)$.

The equations of motion \eqref{sp_prdot}-\eqref{sp_zdot} comprise a Hamiltonian system on the symplectic manifold $M$ parameterized by $(r,\theta,z,p_r,p_\theta,p_z)$, where we recall that $z$ and $\theta$ are each $2\pi$-periodic. The symplectic $2$-form is $\omega = - d\vartheta$, where the Liouville $1$-form is
\begin{align*}
\vartheta = \epsilon(p_r\,dr + p_\theta\,d\theta + p_z\,dz) +\sigma\,[\psi(r)\,d\theta - \psi_P(\psi(r))\,dz].
\end{align*}
Symmetry of the screw pinch magnetic field under translations in $z$ and $\theta$ imply  conservation of $z$- and $\theta$-canonical momenta,
\begin{align*}
P_z &=\epsilon\,p_z - \sigma\,\psi_P\\
P_\theta & = \epsilon\,p_\theta + \sigma\,\psi.
\end{align*}
The three conserved quantities $(P_z,P_\theta,H)$, with $H = \epsilon^2\,(p_r^2 + r^{-2}\,p_\theta^2 + p_z^2)/2$ are functionally-independent for $\epsilon$ nonzero and commute under Poisson bracket. It follows that the system is Liouville integrable.

To describe the foliation by invariant tori it is useful to introduce alternative constants of motion that are well-behaved as $\epsilon\rightarrow 0$:
\begin{align}
\Psi & = \sigma\,P_\theta = \psi + \epsilon\,\sigma\,p_\theta,\label{sp_Psi}\\
P_\parallel & = (\sigma\epsilon)^{-1}(\psi_P(\sigma\,P_\theta) + \sigma\,P_z ) = p_z + p_\theta\,\overline{\iota}(\psi,p_\theta),\label{sp_Pparallel}\\
E & = \epsilon^{-2}H = \frac{1}{2}(p_r^2 + r^{-2}\,p_\theta^2 + p_z^2),\label{sp_H}
\end{align}
where we have introduced the compact notation
\begin{align*}
\overline{\iota}(\psi,p_\theta)& = \int_0^1\iota(\psi + \lambda\epsilon\sigma p_\theta)\,d\lambda.
\end{align*}
Note that $\lim_{\epsilon\rightarrow 0}\overline{\iota}(\psi,p_\theta) = \iota(\psi)$. The limiting forms of these constants of motion as $\epsilon\rightarrow 0$ are given by
\begin{align*}
\lim_{\epsilon\rightarrow 0}\Psi &=  \psi(r) ,\\
\lim_{\epsilon\rightarrow 0}P_\parallel & = p_z + p_\theta\,\iota(\psi(r)),\\
\lim_{\epsilon\rightarrow 0}E & =  \frac{1}{2}(p_r^2 + r^{-2}\,p_\theta^2 + p_z^2).
\end{align*}
Thus, $(\Psi,P_\parallel,E)$ remain functionally-independent even when $\epsilon=0$.
It is also helpful to change coordinates on phase space from $(r,\theta,z,p_r,p_\theta,p_z)$ to $(r,\theta,z,p_\perp,\zeta,p_z)$, where
\begin{align}
p_r &= p_\perp\cos\zeta\label{sp_polar_1}\\
p_\theta & = r^2\iota(\psi)p_z - r\sqrt{1 + r^2\,\iota^2(\psi)}p_\perp\sin\zeta.\label{sp_polar_2}
\end{align}
When $\epsilon = 0$, the invariant torus with constants of motion $(\Psi,P_\parallel,E)$ is parameterized explicitly by $\mathbb{T}^3\rightarrow M:(\theta,z,\zeta)\mapsto \Gamma_0(\theta,z,\zeta\mid \Psi,P_\parallel,E)$, where
\begin{align*}
\Gamma_0(\theta,z,\zeta\mid \Psi,P_\parallel,E) = \begin{pmatrix} r^*_0(\theta,z,\zeta\mid \Psi,P_\parallel,E) \\ \theta \\ z\\ p_{\perp 0}^*(\theta,z,\zeta\mid \Psi,P_\parallel,E)\\ \zeta\\p_{z0}^*(\theta,z,\zeta\mid \Psi,P_\parallel,E)\end{pmatrix},
\end{align*}
where the component functions are 
\begin{align*}
    r^*_0(\theta,z,\zeta\mid \Psi,P_\parallel,E) & = \hat{r}(\Psi)\\
    p_{\perp 0}^*(\theta,z,\zeta\mid \Psi,P_\parallel,E) & = \bigg(2E - \frac{P_\parallel^2}{1 + \hat{r}^2(\Psi)\iota^2(\Psi)}\bigg)^{1/2}\\
    p_{z0}^*(\theta,z,\zeta\mid \Psi,P_\parallel,E) & = \frac{P_\parallel}{1+\hat{r}^2(\Psi)\iota^2(\Psi)} + \frac{\iota(\Psi)\hat{r}(\Psi)}{\sqrt{1 + \hat{r}^2(\Psi)\iota^2(\Psi)}}\bigg(2E - \frac{P_\parallel^2}{1 + \hat{r}^2(\Psi)\iota^2(\Psi)}\bigg)^{1/2}\,\sin\zeta,
\end{align*}
and $\hat{r}= \psi^{-1}$. When $\epsilon $ is non-zero, but small, the torus is instead parameterized by $\mathbb{T}^3\rightarrow M:(\theta,z,\zeta)\mapsto \Gamma(\theta,z,\zeta\mid \Psi,P_\parallel,E)$, where
\begin{align*}
\Gamma(\theta,z,\zeta\mid \Psi,P_\parallel,E) = \begin{pmatrix} r^*(\theta,z,\zeta\mid \Psi,P_\parallel,E) \\ \theta \\ z\\ p_{\perp }^*(\theta,z,\zeta\mid \Psi,P_\parallel,E)\\ \zeta\\p_{z}^*(\theta,z,\zeta\mid \Psi,P_\parallel,E)\end{pmatrix},
\end{align*}
where the functions $r^*,p_\perp^*,p_z^*$ are small perturbations of their limits $r^*_0,p_{\perp 0}^*,p_{z0}^*$. These functions can be computed numerically as follows. 


We wish to relate the phase space variables $(r,\theta,z,p_\perp,\zeta,p_z)$ on the invariant torus with constants of motion $\Psi,P_\parallel,E$. By the definition \eqref{sp_polar_1}-\eqref{sp_polar_2} of $(p_\perp,\zeta)$ we have
\begin{align*}
    \bigg(\frac{r^{-1}p_\theta - r\iota p_z}{\sqrt{1 + r^2\iota^2}}\bigg)^2 =p_\perp^2\sin^2\zeta.
\end{align*}
By energy conservation we have $2E = p_\perp^2\cos^2\zeta + (r^{-1}p_\theta)^2 + p_z^2$. Summing these two relations and rearranging terms implies
\begin{align*}
    p_\perp^2 = 2E + \bigg(\frac{r^{-1}p_\theta - r\iota p_z}{\sqrt{1 + r^2\iota^2}}\bigg)^2 - (r^{-1}p_\theta)^2 - p_z^2.
\end{align*}
By $P_\parallel$-conservation the $z$-momentum may be written $p_z = P_\parallel - \overline{\iota}p_\theta$. Substituting this result into the above expression for $p_\perp^2$ and completing the square implies
\begin{align*}
    p_\perp^2 &= 2E - \frac{P_\parallel^2}{1+r^2\iota^2} - \frac{r^2(\overline{\iota}-\iota)^2}{1 + r^2\iota^2}(r^{-1}p_\theta)^2 + 2\frac{r(\overline{\iota}-\iota)}{1 + r^2\iota^2}(r^{-1}p_\theta)P_\parallel\nonumber\\
    & = 2E - \frac{(P_\parallel - [\overline{\iota}-\iota]p_\theta)^2}{1 + r^2\iota^2},
\end{align*}
which expresses $p_\perp$ as a function of $r,P_\parallel,E,$ and $p_\theta$ on the invariant torus. Again recalling the definition \eqref{sp_polar_1}-\eqref{sp_polar_2} of $(p_\perp,\zeta)$, and in light of $\Psi$-conservation, we have therefore shown $p_\theta = \pi_\theta(\zeta\mid \Psi,P_\parallel,E)$, where $\pi_\theta$ is the unique solution of the fixed point problem
\begin{align}
    \pi_\theta(\zeta\mid \Psi,P_\parallel,E) = \Pi(\pi_\theta(\zeta\mid \Psi,P_\parallel,E)\mid \zeta,\Psi,P_\parallel,E ).\label{sp_fixedpoint}
\end{align}
Here the fixed point map $\Pi$ is given explicitly by
\begin{gather*}
    \Pi(p\mid\zeta,\Psi,P_\parallel,E ) = \hat{r}\frac{\hat{r}\iota P_\parallel - \sqrt{(1+\hat{r}^2\iota^2)2E - (P_\parallel - [\overline{\iota}-\iota]p)^2}\sin\zeta}{1+\hat{r}^2\iota\overline{\iota}}\nonumber\\
    \hat{r} = \hat{r}(\Psi - \epsilon\sigma p),\quad \iota = \iota(\Psi - \epsilon\sigma p),\quad \overline{\iota} = \overline{\iota}(\Psi-\epsilon\sigma p,p).
\end{gather*}
The fixed point $\pi_\theta(\zeta\mid \Psi,P_\parallel,E)$ may be computed rapidly numerically by iterating $\Pi$ on the initial guess
\begin{align*}
   \pi_{\theta}^{(0)}(\zeta\mid \Psi,P_\parallel,E) = \frac{\hat{r}^2\iota P_\parallel}{1 + \hat{r}^2\iota^2} - \frac{\hat{r}}{\sqrt{1 + \hat{r}^2\iota^2}}\sqrt{2E - \frac{P_\parallel^2}{1+\hat{r}^2\iota^2}}\sin\zeta.
\end{align*}
In other words, the sequence defined recursively by \[\pi_\theta^{(k)}(\zeta\mid \Psi,P_\parallel,E) = \Pi(\pi_\theta^{(k-1)}(\zeta\mid \Psi,P_\parallel,E)\mid \zeta,\Psi,P_\parallel,E)\] converges rapidly to the fixed point as $k$ increases. Once the value of $\pi_\theta = \pi_\theta(\zeta\mid \Psi,P_\parallel,E)$ is known, the desired functions $r^*,p_\perp^*,p_z^*$ are given by the explicit formulas
\begin{align}
    r^*(\theta,z,\zeta\mid \Psi,P_\parallel,E)& = \hat{r}(\Psi - \epsilon\sigma\pi_\theta)\label{sp_rstar}\\
    p_\perp^*(\theta,z,\zeta\mid \Psi,P_\parallel,E) & =\sqrt{2E - \frac{(P_\parallel-[\overline{\iota}-\iota]\pi_\theta)^2}{1+\hat{r}^2\iota^2}} \label{sp_pperp}\\
    p_z^*(\theta,z,\zeta\mid \Psi,P_\parallel,E) & =P_\parallel - \overline{\iota}\pi_\theta, \label{sp_pz}
\end{align}
where $\hat{r} = \hat{r}(\Psi-\epsilon\sigma \pi_\theta)$, $\iota = \iota(\Psi-\epsilon\sigma \pi_\theta)$, and $\overline{\iota} = \overline{\iota}(\Psi-\epsilon\sigma \pi_\theta,\pi_\theta)$.

A homology basis for the 3-torus $\Gamma(\Psi,P_\parallel,E)$ is given by the curves $\gamma_i: S^{1}\rightarrow \Gamma(\Psi,P_\parallel, E)$ for $i=1,2,3$, where
\begin{align*}
    \gamma_1(\zeta\mid \Psi,P_\parallel,E) = \Gamma(0,0,\sigma\zeta \mid \Psi,P_\parallel,E)\\
    \gamma_2(z\mid \Psi,P_\parallel,E) = \Gamma(0,z,0\mid \Psi,P_\parallel,E)\\
    \gamma_3(\theta\mid \Psi,P_\parallel,E) = \Gamma(\theta,0,0\mid \Psi,P_\parallel,E).
\end{align*}
(We multiply $\zeta$ by $\sigma$ as a conventional choice only; this choice ensures the loop $\gamma_1$ has the same orientation as the limiting cyclotron orbits.) Integrating the Liouville 1-form around each curve and normalizing gives the action variables $J_i(\Psi,P_\parallel,E)=(2\pi)^{-1}\oint_{\gamma_i}\vartheta$. Doing this for $J_2$ and $J_3$, we recover $P_z$ and $P_\theta$, respectively. For $J_1$, we find
\begin{align*}
    J_1 &= \frac{\epsilon}{2\pi}\int_{\gamma_1}p_rdr\\
    & = \frac{\epsilon\sigma}{2\pi}\int_0^{2\pi}p_\perp^*\,\partial_\zeta r^*\,\cos\zeta\,d\zeta\\
    & = -\frac{\epsilon^2}{2\pi}\int_0^{2\pi}\left({2E - \frac{(P_\parallel-[\overline{\iota}-\iota]\pi_\theta)^2}{1+\hat{r}^2\iota^2}}\right)^{1/2}\hat{r}^\prime(\Psi - \epsilon\sigma\pi_\theta)\,\partial_\zeta\pi_\theta\,\cos\zeta\,d\zeta.
\end{align*}
The factor of $\sigma$ on the second line appears after changing integration variables from $\zeta$ to $\sigma\zeta$. The derivative $\partial_\zeta\pi_\theta$ may be computed in terms of the values of $\pi_\theta$ by implicitly differentiating Eq.\,\eqref{sp_fixedpoint}. The result is
\begin{align*}
    \partial_\zeta\pi_\theta = \frac{\partial_\zeta\Pi(\pi_\theta\mid\zeta, \Psi,P_\parallel,E)}{1-\partial_p\Pi(\pi_\theta\mid\zeta, \Psi,P_\parallel,E)}.
\end{align*}
We conclude that $J_1$ is given by the integral
\begin{align}
    J_1(\Psi,P_\parallel,E)  = -\frac{\epsilon^2}{2\pi}\int_0^{2\pi}\left({2E - \frac{(P_\parallel-[\overline{\iota}-\iota]\pi_\theta)^2}{1+\hat{r}^2\iota^2}}\right)^{1/2}\hat{r}^\prime(\Psi - \epsilon\sigma\pi_\theta)\,\frac{\partial_\zeta\Pi}{1 - \partial_p\Pi}\,\cos\zeta\,d\zeta,\label{sp_J1}
\end{align}
where $\iota = \iota(\Psi - \epsilon\sigma\pi_\theta)$, $\overline{\iota} = \overline{\iota}(\Psi-\epsilon\sigma\pi_\theta,\pi_\theta)$, $\pi_\theta = \pi_\theta(\zeta\mid \Psi,P_\parallel,E)$, $\partial_\zeta\Pi = \partial_\zeta\Pi(\pi_\theta\mid \zeta,\Psi,P_\parallel,E)$, and $\partial_p\Pi = \partial_p\Pi(\pi_\theta\mid \zeta,\Psi,P_\parallel,E)$.

We claim that the first action variable $J_1$ is the non-perturbative adiabatic invariant for this magnetic field. We argue precisely as follows.
\begin{proposition}\label{sp_adinv_prop}
    The constant of motion $\mathcal{J} = J_1$, given in Eq.\,\eqref{sp_J1}, is a non-perturbative adiabatic invariant for Eqs.\,\eqref{sp_prdot}-\eqref{sp_zdot}. In particular, the series expansion of $J_1$ in powers of $\epsilon$ agrees with Kruskal's adiabatic invariant series to all orders in $\epsilon$.
\end{proposition}
\begin{proof}
The proof proceeds exactly as in the proof of Prop. \ref{slab_adinv_prop}. In particular, since we already know that $X_{J_1}$ generates a circle action that leaves Lorentz force dynamics invariant, it is enough to show that $\lim_{\epsilon\rightarrow 0}X_{J_1} = \mathcal{R}_0$, where $\mathcal{R}_0$ denotes the limiting roto-rate,
\begin{align*}
    \mathcal{R}_0 = \frac{1}{B(r)}\lim_{\epsilon\rightarrow 0}X_{H} = \frac{\sigma}{r^{-1}\sqrt{1+r^2\iota^2}}\bigg([r^{-2}p_\theta - \iota p_z]\partial_r- p_r\partial_\theta + \iota\,p_r\,\partial_z\bigg).
\end{align*}
Here we have used a convenient formula for the magnetic field strength, $B(r) = r^{-1}\sqrt{1+r^2\iota^2}\,\psi^\prime$.

The first two non-vanishing terms in the series expansion for $J_1$ are given by 
\begin{align}
	J_1(\Psi,P_\parallel,E)&=\epsilon^2\frac{ \hat r \hat r' }{2\sqrt{1+\hat r^2 \iota^2 }}\left(2E-\frac{P_\parallel^2}{1+\hat r^2\iota^2}\right) \nonumber\\
    &+\sigma\epsilon^3\frac{P_\parallel \hat r^2 }{4(1+\hat r^2 \iota^2 )^{7/2}}\bigg(-2\hat r'^2 \iota (6E-3P_\parallel^2+2(3E+P_\parallel^2)\hat r^2 \iota^2 )\nonumber\\
    &-2\hat r \iota (1+\hat r^2 \iota^2 )(2E[1+\hat r^2 \iota^2 ]-P_\parallel^2)\hat r'' \nonumber\\
    &+\hat r \hat r' (-2E+P_\parallel^2+2(E-2P_\parallel^2)\hat r^2 \iota^2 +4E\hat r^4 \iota^4 )\iota' \bigg)+O(\epsilon^4),\label{sp_J1_aa}
\end{align}
where $\hat{r} = \hat{r}(\Psi)$ and $\iota = \iota(\Psi)$.
The leading-order derivatives of $J_1$ are therefore  
\begin{align*}
    \partial_E J_1(\Psi,P_\parallel,E)&=\epsilon^2\frac{\hat{r}\hat{r}'}{\sqrt{1+\hat{r}^2 \iota^2}}+O(\epsilon^3)=\epsilon^2\frac{1}{B(\hat{r})} + O(\epsilon^3)\\
    \partial_\Psi J_1(\Psi,P_\parallel,E)&=\epsilon^2\frac{1}{2(1+\hat{r}^2 \iota^2)^{5/2}} \bigg((2E-P_\parallel^2+2(E+P_\parallel^2)\hat r^2\iota^2)\hat r '^2\\
            &+\hat r^3\iota(3P_\parallel^2-2E(1+\hat r^2\iota^2))\hat r'\iota'\\
            &-\hat r(1+\hat r^2\iota^2)(P_\parallel^2-2E(1+\hat r^2\iota ^2)\hat r'' \bigg)+O(\epsilon^3)\\
    \partial_{P_\parallel} J_1(\Psi,P_\parallel,E)&=-\epsilon^2\frac{ P_\parallel \hat r\hat r'}{(1+\hat r^2\iota^2)^{3/2}}+O(\epsilon^3).
\end{align*}
These derivatives enable the following computation of $\lim_{\epsilon\rightarrow 0}X_{J_1}$:
\begin{align*}
    \lim_{\epsilon\rightarrow 0}X_{J_1} & = \lim_{\epsilon\rightarrow 0}\bigg(\partial_{\Psi}J_1X_\Psi+\partial_{P_\parallel}J_1 X_{P_\parallel} + \partial_{E}J_1 X_{E} \bigg)\\
    & = \lim_{\epsilon\rightarrow 0}\bigg(\partial_{\Psi}J_1(\sigma\partial_\theta)+\frac{1}{\epsilon}\partial_{P_\parallel}J_1 (\iota\,\partial_\theta + \partial_z) + \frac{1}{\epsilon^2}\partial_{E}J_1 X_{H} \bigg)\\
    & = \lim_{\epsilon\rightarrow 0}\bigg( \frac{1}{\epsilon^2}\partial_{E}J_1 X_{H} \bigg)\\
    & = \frac{1}{B(r)}\lim_{\epsilon\rightarrow 0}X_H = \mathcal{R}_0,
\end{align*}
which is the desired result.

\end{proof}

\subsection{Direct comparison with Kruskal's series}
When $\epsilon\ll 1$, which corresponds to the asymptotic regime where traditional guiding center theory applies, it is sensible to expand Eq.\,\eqref{sp_J1} in powers of $\epsilon$. The first two non-vanishing terms in the series are recorded in Eq.\,\eqref{sp_J1_aa}, $J_1(\Psi,P_\parallel,E) = \epsilon^2\,J_{12}(\Psi,P_\parallel,E) + \epsilon^3\,J_{13}(\Psi,P_\parallel,E)+...$. Substituting Eqs.\,\eqref{sp_Psi}-\eqref{sp_H} into these formulas then leads to the first two non-vanishing terms of the nonperturbative invariant $\mathcal{J} = \epsilon^2\mathcal{J}_2 + \epsilon^3\mathcal{J}_3 + \dots$ as functions of $(r,\theta,z,p_r,p_\theta,p_z)$. In light of Prop. \ref{sp_adinv_prop}, these terms must agree with the known explicit expressions for the first two terms in Kruskal's adiabatic invariant series in general magnetic geometries, usually denoted $\mu_0$ and $\mu_1$. In order to emphasize the power of Prop.\,\ref{sp_adinv_prop}, we will now explicitly compare $\mu_0$ with $\mathcal{J}_2$ and $\mu_1$ with $\mathcal{J}_3$, using the formulas in Refs. \onlinecite{weyssow_hamiltonian_1986,burby_automation_2013} for $\mu_0,\mu_1$.


The expressions for the first two terms in Kruskal's adiabatic invariant series are usually expressed in terms of $\bm{B}$ and its derivatives, together with components of the particle velocity $\bm{v}$. In this case the magnetic field is\[
\bm B=\frac{1}{r}\iota \psi'\partial_\theta +\frac{1}{r}\psi'\partial_z
\] and the metric tensor is $g = dr^2 + r^2d\theta^2 + dz^2$. As mentioned in the proof of Prop.\,\ref{sp_adinv_prop},  the magnetic field strength is therefore \[
|\bm B|^2 = g(\bm B,\bm B)=r^{-2}(1+r^2\iota^2)(\psi')^2.
\] The unit vector in the magnetic field direction is \[
\bm b=\bm B/|\bm B|=\frac{\iota}{\sqrt{1+r^2\iota^2}}\partial_\theta+\frac{1}{\sqrt{1+r^2\iota^2}}\partial_z.
\]
Furthermore, the particle velocity is \[
\bm v= p_r\partial_r+ r^{-2}p_\theta \partial_\theta+ p_z\partial_z.
\] The dot and cross products with $\bm b$ are 
\begin{align*}
\bm v\cdot\bm b &= \frac{p_z+\iota p_\theta}{\sqrt{1+r^2\iota^2}}\\
\bm v\times \bm b&=\frac{1}{\sqrt{1+r^2\iota^2}}\left((r^{-1}p_\theta-p_z\iota r)\partial_r- p_r r^{-1}\partial_\theta+p_r\iota r\partial_z\right).
\end{align*} 

For $\mu_0 = |\bm{v}\times\bm{b}|^2/(2|\bm{B}|)$, we need 
\[
|\bm v\times\bm b|^2=\left(p_r^2+\frac{(p_\theta-\iota p_zr^2)^2}{r^2(1+r^2\iota^2)}\right).
\] 
Thus, in $(r,\theta,z,p_r,p_\theta,p_z)$-coordinates, 
\[
\mu_0=\frac{|\bm v\times\bm b|^2}{2|\bm B|}=\frac{r }{2\psi'(r)\sqrt{1+r^2\iota^2}}\left(p_r^2+\frac{(p_\theta-\iota p_z r^2)^2}{r^2(1+r^2\iota^2)}\right).
\] 
This expression agrees with the leading-order term in Eq.\,\eqref{sp_J1_aa} when expressed in terms of $(r,\theta,z,p_r,p_\theta,p_z)$, as expected.

Next we compute the first correction,  
\begin{align*}
\mu_1 &=\mu_0\frac{(\bm b\times \bm v)\cdot\nabla|\bm B|}{|\bm B|^2}+\frac{1}{4}\frac{(\bm v\cdot\bm b)\bm v\cdot\nabla\bm b\cdot(\bm v\times\bm b)}{|\bm B|^2}\\
&-\frac{3}{4}\frac{(\bm v\cdot\bm b)(\bm v\times\bm b)\cdot\nabla\bm b\cdot\bm v}{|\bm B|^2}-\frac{5}{4}\frac{(\bm v\cdot\bm b)^2\bm\kappa\cdot (\bm v\times\bm b)}{|\bm B|^2},
\end{align*}
where $\bm\kappa=\bm b\cdot\nabla\bm b$. See Eq.\,(29) in Ref.\,\onlinecite{burby_automation_2013}. For general $\bm V=V^r\partial_r+V^\theta\partial_\theta+V^z\partial_z$, we have \[
\bm V\cdot\nabla \bm b=V_r\left(\frac{\partial b^\theta}{\partial r}+\frac{1}{r}b^\theta\right)\partial_\theta+V_r\frac{\partial b^z}{\partial r}\partial_z-V^\theta rb^\theta\partial_r.
\] 
Using a computer algebra system it is straightforward to find 
\begin{align*}
\mu_1&=-\frac{2r\psi''}{4 r^2 (r^2
   \iota^2+1)^{7/2} (\psi')^3}\bigg( (r^2 \iota^2+1)  (p_\theta-p_z r^2 \iota) (r^4
   (p_r^2+p_z^2) \iota^2+p_r^2
   r^2+p_\theta^2-2 p_\theta p_z r^2 \iota)\bigg)\\
   &-\frac{2\psi'}{4 r^2 (r^2
   \iota^2+1)^{7/2} (\psi')^3}\bigg(
    r^2 \iota \bigg[p_z (3
   p_r^2 r^2+5 p_\theta^2)\\
   &+\iota [r^2 \iota
    \{2 p_\theta \iota (r^2
   (p_r^2+3 p_z^2)-p_\theta^2+p_\theta p_z
   r^2 \iota)+p_z r^2 (3 p_r^2+5
   p_z^2)\\
   &-8 p_\theta^2 p_z\}+p_\theta r^2
   (p_r^2-9 p_z^2)+2
   p_\theta^3]\bigg]-p_\theta (p_r^2
   r^2+p_\theta^2)\bigg)\\
   &-\frac{r^3 (\psi')^2 \iota'}{4 r^2 (r^2
   \iota^2+1)^{7/2} (\psi')^3}\bigg(
   p_z (p_r^2 r^2+3 p_\theta^2)\\
   &+\iota  \bigg[r^2 \iota [3 p_\theta r^2
   (p_r^2+3 p_z^2) \iota-2 p_z r^4
   (p_r^2+p_z^2) \iota^2\\
   &+p_z r^2
   (3 p_z^2-p_r^2)-12 p_\theta^2
   p_z]+3 p_\theta r^2 (p_r^2-2
   p_z^2)+5 p_\theta^3\bigg]\bigg).
\end{align*}
We have also used a computer algebra system confirm that this expression agrees with the coefficient in front of $\epsilon^3$ in Eq.\,\eqref{sp_J1_aa} when written in terms of $(r,\theta,z,p_r,p_\theta,p_z)$.

\subsection{Nonperturbative guiding center model}
We may now identify the non-perturbative guiding center equations of motion, as formulated in Ref.\,\onlinecite{j_w_burby_nonperturbative_2025} and the associated supplemental material, for the screw pinch. By using $\mathcal{J} = J_1$ as the non-perturbative adiabatic invariant the guiding center model we derive in this manner should make predictions that agree exactly with those of the full-order Lorentz force model. If we use a truncated power series expansion of $J_1$ for $\mathcal{J}$ the model should instead agree with the traditional guiding center model, truncated at some order in $\epsilon$. 

The nonperturbative guiding center formalism from Ref.\,\cite{j_w_burby_nonperturbative_2025} requires two inputs: a Poincar\'e section $\Sigma\subset M$ for gyromotion and an expression for the adiabatic invariant $\mathcal{J}$. The Poincar\'e section serves as the $5D$ guiding center phase space. We will work in the coordinates $(r,\theta,z,p_\perp,\zeta,p_z)$ on $M$ and define $\Sigma = \{\zeta = 0\}$, so that $(r,\theta,z,p_\perp,p_z)$ provides a simple parameterization of $\Sigma$. We will also refer to the notation developed in Ref.\,\cite{j_w_burby_nonperturbative_2025} and the associated supplemental material. 

The Hamiltonian restricted to $\Sigma$ is 
\[
H_\Sigma=H|_\Sigma=\frac{\epsilon^2}{2}\left(p_\perp^2+[1+r^2\iota^2(\psi)]p_z^2\right).
\] 
To compute the Poisson bracket $\{\cdot,\cdot\}_\Sigma$ (see Theorem 2 in the supplemental material), we need to compute the pairwise Poisson brackets of the coordinates $(r,\theta,z,p_\perp,\zeta,p_z)$. We do this by computing the symplectic form in these coordinates and then computing the Hamiltonian vector fields of the coordinate functions from that. First compute differentials of $p_\theta$ and $p_r$: \[
dp_r=\cos\zeta dp_\perp-p_\perp\sin\zeta d\zeta,
\]
\begin{align*}
    dp_\theta &= d(r^2\iota(\psi)p_z-rp_\perp\sin\zeta\sqrt{1+r^2\iota(\psi)^2})\\
    &= u(r,p_\perp,\zeta,p_z)dr+r^2\iota(\psi)dp_z-r\sin\zeta\sqrt{1+r^2\iota(\psi)^2}dp_\perp-r p_\perp\cos\zeta\sqrt{1+r^2\iota(\psi)^2}d\zeta,
\end{align*}
where \[
u(r,p_\perp,\zeta,p_z)=[2r\iota(\psi)+r^2\psi'\iota'(\psi)]p_z-p_\perp\sin\zeta\left(\sqrt{1+r^2\iota(\psi)^2}+\frac{r^2\iota(\psi)(\iota(\psi)+r\iota'(\psi)\psi'}{\sqrt{1+r^2\iota(\psi)^2}})\right).
\]
The symplectic form in these coordinates is thus
\begin{align*}
\omega=-d\vartheta&=\epsilon(\cos\zeta dr\wedge dp_\perp-p_\perp\sin\zeta dr\wedge d\zeta+u(r,p_\perp,\zeta,p_z)d\theta\wedge dr+r^2\iota(\psi)d\theta\wedge dp_z\\
&-r\sin\zeta\sqrt{1+r^2\iota(\psi)^2}d\theta\wedge dp_\perp-rp_\perp\cos\zeta\sqrt{1+r^2\iota(\psi)^2}d\theta\wedge d\zeta+dz\wedge dp_z)\\
&-\sigma\psi' (dr\wedge d\theta-\iota(\psi)dr\wedge dz).
\end{align*}
By inverting the matrix associated with $\omega$ we obtain the Hamiltonian vector fields associated with the coordinate functions:
\begin{align*}
    X_r&=-\epsilon^{-1}\cos\zeta \partial_{p_\perp}+\epsilon^{-1}p_\perp^{-1}\sin\zeta \partial_{\zeta}\\
    X_\theta&=\epsilon^{-1}\frac{r^{-1}\sin\zeta}{\sqrt{1+r^2\iota^2}}\partial_{p_\perp}+\epsilon^{-1}\frac{r^{-1}p_\perp^{-1}\cos\zeta}{\sqrt{1+r^2\iota^2}}\partial_\zeta\\
    X_z&=-\epsilon^{-1}\frac{r\iota\sin\zeta}{\sqrt{1+r^2\iota^2}}\partial_{p_\perp}-\epsilon^{-1}\frac{r\iota p_\perp^{-1}\cos\zeta}{\sqrt{1+r^2\iota^2}}\partial_\zeta-\epsilon^{-1}\partial_{p_z}\\
    X_{p_\perp}&=\epsilon^{-1}\cos\zeta\partial_r-\epsilon^{-1}\frac{r^{-1}\sin\zeta}{\sqrt{1+r^2\iota^2}}\partial_\theta+\epsilon^{-1}\frac{r\iota\sin\zeta}{\sqrt{1+r^2\iota^2}}\partial_z\\
    &+\epsilon^{-2}\frac{r^{-1}p_\perp^{-1}(\epsilon u+\sigma\psi'(1+r^2\iota^2))}{\sqrt{1+r^2\iota^2}}\partial_\zeta+\epsilon^{-2}\sigma\iota\psi'\cos\zeta\partial_{p_z}\\
    X_\zeta&=-\epsilon^{-1}p_\perp^{-1}\sin\zeta\partial_r-\epsilon^{-1}\frac{r^{-1}p_\perp^{-1}\cos\zeta}{\sqrt{1+r^2\iota^2}}\partial_\theta+\epsilon^{-1}\frac{r\iota p_\perp^{-1}\cos\zeta}{\sqrt{1+r^2\iota^2}}\partial_z\\
    &-\epsilon^{-2}\frac{r^{-1}p_\perp^{-1}(\epsilon u+\sigma\psi'(1+r^2\iota^2))}{\sqrt{1+r^2\iota^2}}\partial_{p_\perp}-\epsilon^{-2}\sigma\iota p_\perp^{-1}\psi \sin\zeta\partial_{p_z}\\
    X_{p_z}&=\epsilon^{-1}\partial_z-\epsilon^{-2}\sigma\iota\psi'\cos\zeta\partial_{p_\perp}+\epsilon^{-2}\sigma\iota p_\perp^{-1}\iota\psi'\sin\zeta\partial_{\zeta}.
\end{align*}
The coefficients of these vector fields immediately give the Poisson brackets among the coordinates. The column vector $N_\Sigma$ (again we refer to the supplemental material) is thus \[
N_\Sigma=\frac{p_\perp^{-1}}{\sqrt{1+r^2\iota^2}}\begin{pmatrix}
    0\\
    -\epsilon^{-1}r^{-1}\\
    \epsilon^{-1}r\iota \\
    -\epsilon^{-2}r^{-1}(\epsilon u|_{\zeta=0}+\sigma\psi'(1+r^2\iota^2))\\
    0\\
\end{pmatrix},
\] and the matrix $\mathbb J_\Sigma$ is
\[
\mathbb J_\Sigma=\begin{pmatrix}
    0&0&0&\epsilon^{-1}&0\\
    0&0&0&0&0\\
    0&0&0&0&\epsilon^{-1}\\
    -\epsilon^{-1}&0&0&0&-\epsilon^{-2}\sigma\iota\psi'\\
    0&0&-\epsilon^{-1}&\epsilon^{-2}\sigma\iota\psi'&0
\end{pmatrix}.
\] 
We also find that the partial derivatives of $H_\Sigma$ are given by
\[
\partial_\sigma H_\Sigma=
\epsilon^2
\begin{pmatrix}
rp_z^2\iota(\iota+r\iota'\psi')\\
0\\
0\\
p_\perp\\
p_z(1+r^2\iota^2)
\end{pmatrix}.
\]
This implies 
\begin{align*}
\frac{N_\Sigma^T\partial_\sigma H_\Sigma}{N_\Sigma^T\partial_\sigma J_\Sigma}=\frac{\epsilon^2 p_\perp}{\partial_{p_\perp} J_\Sigma},\partial_\sigma J_\Sigma^T\mathbb J_\Sigma\partial_\sigma H_\Sigma&=\sigma\iota\psi'(p_\perp\partial_{p_z}J_\Sigma-p_z(1+r^2\iota^2)\partial_{p_\perp}J_\Sigma)\\
&+\epsilon (p_\perp\partial_r J_\Sigma-rp_z^2\iota(\iota+r\iota'\psi')\partial_{p_\perp}J_\Sigma),
\end{align*}
where we've used the fact that $J_\Sigma$ is independent of $\theta$ and $z$ to drop some terms. We can now write down the non-perturbative equations of motion by computing Poisson brackets $\{z,H_\Sigma\}_\Sigma$, where $z$ is any of $(r,z,\theta,p_\perp,p_z)$:
\begin{align}
    \dot r &= 0\label{nonpertgc1}\\
    \dot\theta&=-\frac{\epsilon (\partial_{p_\perp}J_\Sigma)^{-1}}{\epsilon u +\sigma\psi'(1+r^2\iota^2)}\partial_\sigma J_\Sigma^T\mathbb J_\Sigma\partial_\sigma H_\Sigma\label{nonpertgc2}\\
    \dot z&=\epsilon p_z(1+r^2\iota^2)-\epsilon p_\perp\frac{\partial_{p_z}J_\Sigma}{\partial_{p_\perp}J_\Sigma}+\frac{\epsilon r^2\iota(\partial_\sigma J_\Sigma^T\mathbb J_\Sigma\partial_\sigma H_\Sigma)(\partial_{p_\perp}J_\Sigma)^{-1}}{\epsilon u+\sigma\psi'(1+r^2\iota^2)}\label{nonpertgc3}\\
    \dot p_\perp&=0\label{nonpertgc4}\\
    \dot p_z&=0.\label{nonpertgc5}
\end{align}
Note that we should have anticipated constancy of $r,p_\perp,p_z$ in the nonperturbative guiding center model because the constants of motion $(\Psi,P_\parallel,E)$ are each functions of $(r,p_\perp,p_z)$ when restricted to $\Sigma$.
\subsection{Numerical assessment of nonperturbative guiding center theory}
We will now assess the predictions of nonperturbative guiding center theory for the screw pinch field in three ways. (I) In order to determine when the nonperturbative formalism should be preferred over traditional guiding center theory we will study how the optimal truncation order for the series expansion of $J_1$ varies with $\epsilon$. (II) In order to verify the striking prediction that $X_{J_1}$ is the infinitesimal generator of a cricle action when $J_1$ is given by the full integral expression \eqref{sp_J1}  we will numerically test whether the integral curves of $X_{J_1}$ are $2\pi$-periodic. (III) In order to verify that the nonperturbative guiding center equations of motion agree exactly with the full-orbit model when $\mathcal{J}$ is the exact nonperturbative invariant we will compare the orbital frequencies of full-orbit with those of the nonperturbative model.
\\ \\
\noindent (I) \textbf{Optimal truncation order.} We consider the following setup:
\[
    \Psi=1,P_\parallel=0.5,E=3,\sigma=1
\]
\[
\iota(\psi)=\sqrt{2},\psi(r)=r^2.
\]
Once we fix $\epsilon$, we can numerically compute the exact $J_1$ by iterating the fixed point map to get $\pi_\theta$ and then applying the trapezoid rule. Since the integrand is analytic, periodic, and is being integrated over its period, we get exponential convergence, with on the order of 10-20 fixed point iterates and mesh points needed for machine precision. Explicitly, the first several terms in the asymptotic series for $J_1$ using the formulas above are (to four digits)
\[
J_1(1,0.5,3)=0.8540\epsilon^2-0.0019\epsilon^3-0.0940\epsilon^4-0.0842\epsilon^5+O(\epsilon^6).
\]
\begin{figure}
    \centering
    \includegraphics[width=1.0\linewidth]{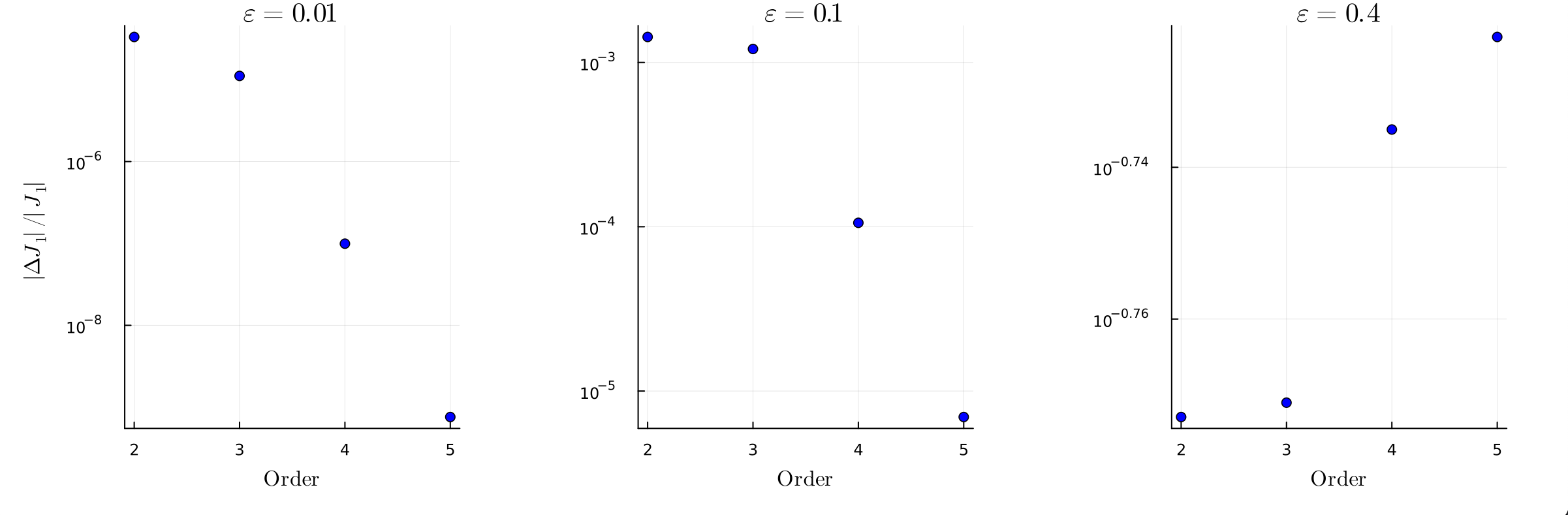}
    \caption{Deviation between exact adiabatic invariant and truncations of its power series in $\epsilon$ vs truncation order.}
    \label{fig:series_collapse}
\end{figure}
We compare these to the numerical $J_1$ in Fig. \ref{fig:series_collapse} ($\Delta J_1$ is the difference). Note that for small $\epsilon$ we see each term improves accuracy, confirming the accuracy of our computations. However, once $\epsilon$ exceeds about $.4$ the series expansion begins to poorly approximate the true invariant. Note that the series for this field breaks down at larger $\epsilon_{\text{breakdown}} \sim .4$ than observed in Ref.\,\onlinecite{j_w_burby_nonperturbative_2025}, where $\epsilon_{\text{breakdown}}\sim .15$. This indicates that there is no ``universal" value for $\epsilon$ above which the traditional series expansions break down. We anticipate that $\epsilon_{\text{breakdown}}$ is smaller in magnetic fields with particle orbits that are further from integrability, but this phenomenon deserves further study.
\\ \\
\noindent (II) \textbf{Symmetry periodicity.} In order to compute the Hamiltonian flow and the non-perturbative equations of motion for $J_1(\Psi,P_\parallel,E)$, we first compute $\partial_\Psi J_1,\partial_{P_\parallel}J_1,$ and $\partial_E J_1$ in terms of derivatives of $\pi_\theta$. Using the shorthand 
\[p_\perp=\left(2E -\frac{(P_\parallel - [\bar\iota-\iota]\pi_\theta)^2}{1+\hat r^2\iota^2}\right)^{1/2},\] 
and differentiating under the integral sign, these are
\begin{align*}
\partial_\Psi J_1&=-\frac{\epsilon^2}{2\pi}\int_0^{2\pi}[\partial_\zeta \pi_\theta \hat r' \partial_\Psi p_\perp-p_\perp(\partial_\zeta\pi_\theta \hat r''(-1+\sigma \epsilon \partial_\Psi\pi_\theta)-\hat r'\partial_\Psi\partial_\zeta \pi_\theta)]\cos\zeta d\zeta\\
    \partial_{P_\parallel}J_1&=-\frac{\epsilon^2}{2\pi}\int_0^{2\pi}[\partial_\zeta \pi_\theta(\hat r'\partial_{P_\parallel}p_\perp-\sigma \epsilon p_\perp \hat r''\partial_{P_\parallel}\pi_\theta)+p_\perp\hat r'\partial_{P_\parallel}\partial_{\zeta}\pi_\theta]\cos\zeta d\zeta\\
\partial_E J_1&=-\frac{\epsilon^2}{2\pi}\int_0^{2\pi}[\partial_\zeta\pi_\theta (\hat r'\partial_E p_\perp -\sigma\epsilon p_\perp \hat r''\partial_E\pi_\theta)+p_\perp\hat r'\partial_E\partial_\zeta\pi_\theta]\cos\zeta d\zeta\\
\end{align*}
where 
\begin{align*}
    \partial_\Psi p_\perp&=(P_\parallel +(\iota-\bar\iota)\pi_\theta)(-2\hat r\iota(P_\parallel+(\iota-\bar\iota)\pi_\theta)(\iota\hat r'+\hat r\iota')(-1+\epsilon\sigma \partial_\Psi \pi_\theta)\\
    &-2(1+\hat r^2\iota^2)((\iota-\bar\iota)\partial_\Psi\pi_\theta+\pi_\theta(-\partial_{p_\theta}\bar\iota\partial_\Psi\pi_\theta+\iota'(1-\epsilon\sigma \partial_\Psi\pi_\theta)+\partial_\psi\bar\iota(-1+\epsilon\sigma\partial_\Psi\pi_\theta))))\\
    &/(2(1+\hat r^2\iota^2)^2(2E-(P_\parallel+(\iota-\bar\iota)\pi_\theta)^2/(1+\hat r^2\iota^2))^{1/2})\\
    \partial_{P_\parallel} p_\perp&=(P_\parallel +(\iota-\bar\iota)\pi_\theta)(-2\epsilon\sigma\hat r\iota(P_\parallel +(\iota-\bar\iota)\pi_\theta)(\iota\hat r'+\hat r\iota')\partial_{P_\parallel}\pi_\theta\\
    &-2(1+\hat r^2\iota^2)(1+(\iota-\bar\iota-\pi_\theta(\partial_{p_\theta}\bar\iota+\epsilon\sigma(\iota'-\partial_\psi\bar\iota)))\partial_{P_\parallel}\pi_\theta))\\
    &/(2(1+\hat r^2\iota^2)^2(2E-(P_\parallel+(\iota-\bar\iota)\pi_\theta)^2/(1+\hat r^2\iota^2))^{1/2})
    \\
    \partial_E p_\perp&=(2-(2\epsilon\sigma \hat r\iota(P_\parallel+(\iota-\bar\iota)\pi_\theta)^2(\iota\hat r'+\hat r\iota')\partial_E\pi_\theta)/(1+\hat r^2\iota^2)^2\\
    &-(2(P_\parallel+(\iota-\bar\iota)\pi_\theta)(\iota-\bar\iota-\pi_\theta(\partial_{p_\theta}\bar\iota+\epsilon\sigma(\iota'-\partial_\psi\bar\iota)))\partial_E\pi_\theta/(1+\hat r^2\iota^2)))\\
    &/(2(2E-(P_\parallel+(\iota-\bar\iota)\pi_\theta)^2/(1+\hat r^2\iota^2))^{1/2})
\end{align*}
and $\hat r'=\hat r'(\Psi-\epsilon\sigma\pi_\theta)$. Once we have an accurate value of $\pi_\theta$ from fixed point iteration, we can directly compute its partial derivatives from the partial derivatives of $\Pi$. We compute these numerically via automatic differentiation.\cite{RevelsLubinPapamarkou2016} In our original phase space coordinates, the Hamiltonian flow for $J_1$ is the solution to the ODE system
\begin{align*}
\dot p_r&=\epsilon^{-2}\partial_EJ_1[(\sigma \psi')(r^{-2}p_\theta-\iota(\psi)p_z)+\epsilon r^{-3}p_\theta^2]\\
\dot p_\theta &=-\sigma\epsilon^{-2}\partial_E J_1 p_r\psi'\\
\dot p_z &= \sigma\epsilon^{-2}\partial_E J_1 p_r \iota(\psi)\psi'\\
\dot r&=\epsilon^{-1}\partial_E J_1 p_r\\
\dot \theta&= \epsilon^{-1}(r^{-2}p_\theta\partial_E J_1+\partial_{P_\parallel}J_1\iota(\psi+\sigma \epsilon p_\theta))+\sigma\partial_\Psi J_1\\
\dot z&=\epsilon^{-1}(\partial_E J_1 p_z+\partial_{P_\parallel}J_1)
\end{align*}
We can then check numerically that the Hamiltonian flow of $J_1$ is periodic with period $2\pi$. \\
To simulate the non-perturbative equations of motion, we need derivatives of $J_1$ in terms of the coordinates $(r,\theta,z,p_\perp,p_z)$. We implement the map $(r,\theta,z,p_\perp,p_z)\to(\Psi,P_\parallel,E)$, compute the derivatives we need using automatic differentiation, and apply the chain rule. 

We numerically demonstrate the $2\pi$-periodicity of the Hamiltonian trajectories with Hamiltonian $J_1$ as follows. We sample $N=100$ random initial conditions $\{\xi_i(0)\}_{i=1,\dots, N}$ from the full-orbit phase space. Using a timestep $\Delta t$ we integrate the initial conditions over the time interval $[0,2\pi]$ using a Runge-Kutta scheme applied to Hamilton's equations for $J_1$. This results in a new sequence of points in phase space $\{\xi_i^{\Delta t}(2\pi)\}_{i=1,\dots,N}$ that approximates the time-$2\pi$ evolution of each initial condition under the $J_1$ flow. Due to the truncation error inherent to Runge-Kutta, the absolute differences $||\xi_i^{\Delta t}(2\pi)-\xi_i(0) ||$ may not vanish when the underlying dynamics is actually $2\pi$-periodic. However, for small enough $\Delta t$, the differences should obey a power law scaling $||\xi_i^{\Delta t}(2\pi)-\xi_i(0) || \sim \Delta t^\alpha$, where $\alpha$ denotes the order of the integrator, if the true dynamics is $2\pi$-periodic. In particular the mean log error
\begin{align*}
    MLE(\Delta t) = \frac{1}{N}\sum_{i=1}^N \text{log}_{10}||\xi_i^{\Delta t}(2\pi)-\xi_i(0)||
\end{align*}
should be an affine function of $\text{log}_{10}\Delta t$ with slope  $\alpha$. Figure \ref{fig:integralcurveperiodicity} displays the computed values of $MLE(\Delta t)$ for various values of $\text{log}_{10}\Delta t$. The values lie along a line with slope given by the order of the integrator, as expected.

\begin{figure}[H]
    \centering
    \includegraphics[width=1\linewidth]{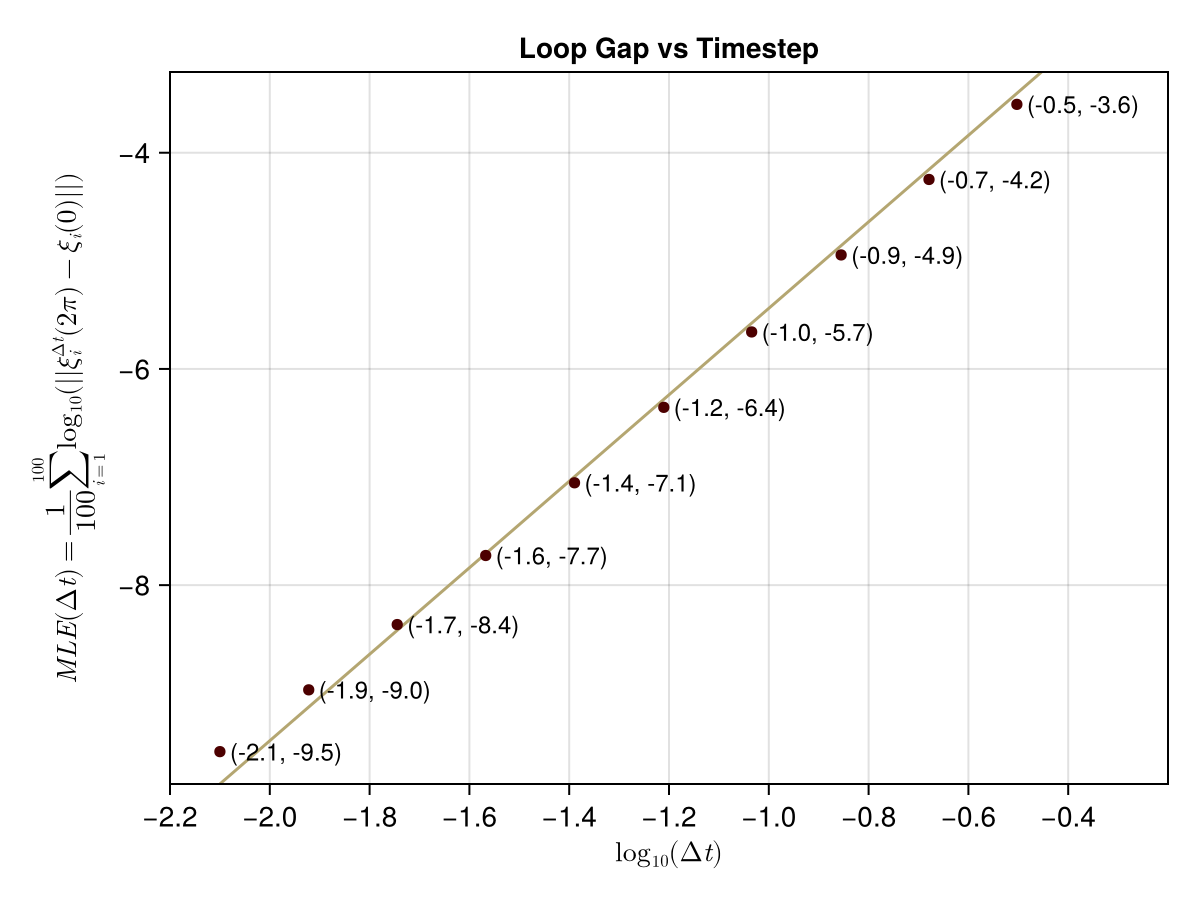}
    \caption{For the Hamiltonian flow $J_1$ with $\epsilon = 0.1$, we compute the average magnitude of the difference $||\xi_i^{\Delta t}(2\pi) - \xi_i(0)||$, where $\xi_i = (r_i, \theta, z, p_r, p_\theta, p_z) 
\in [0.25,1] \times [0, 2\pi] \times [0, 2\pi] \times [-1, 1] \times [-1, 1] \times [-1, 1]$ for $i = 1, ...,100$. This difference, computed using the fourth-order Runge–Kutta (RK4) scheme with varying time steps, exhibits the expected slope of 4.}
    \label{fig:integralcurveperiodicity}
\end{figure}

\noindent (III) \textbf{Nonperturbative predictions.} Here we verify numerically that the nonperturbative equations of motion match the full-orbit equations of motion in the sense that integrating both for a single return time (the time in which a particle returns to the Poincar\'e section) gives the same result. We define a Poincaré section by setting $\zeta = 0$. From Eqs. \eqref{sp_polar_1}-\eqref{sp_polar_2}, this gives 
\begin{align}
p_r &= p_{\perp}, \\
p_{\theta} &= r^2\,\iota(\psi)\,p_z.\label{section_def}
\end{align}
Since $p_{\perp}$ is always positive, we can use \eqref{section_def} and the condition $p_r > 0$ to identify points on this section. We initialize the non-perturbative system \eqref{nonpertgc1}-\eqref{nonpertgc5} at ($r_0$, $\theta_0$, $z_0$, $p_{\perp 0}$, $p_{z0}$) and the full orbit system at ($r_0$, $\theta_0$, $z_0$, $p_{r0}$, $p_{\theta0}$, $p_{z0}$) = ($r_0$, $\theta_0$, $z_0$, $p_{\perp 0}$, $r^2\,\iota(\psi)\,p_{z0}$, $p_{z0}$). In this setup, we used $\iota(\psi)=\sqrt{2}\,,\,\psi(r)=r^2$. Both equations are integrated using the fourth-order Runge–Kutta (RK4) method. For the full-orbit case, we detect crossings with the Poincaré section using a bisection method. We then compute $\Delta z/\Delta t$ for each model, where $\Delta t$ is the first-return time predicted by the full-orbit model, and $\Delta z$ is the predicted change in $z$. By \eqref{nonpertgc2}-\eqref{nonpertgc3} this should compute the angular frequency of rotations along the $z$-axis. Fig. \ref{fig:frequencycomparision} displays this comparison for multiple combinations of $\epsilon$ and initial particle radius $r_0$. Predictions from the two models agree within machine precision, confirming the accuracy of the nonperturbative guiding center model in these fields.

\begin{figure}[H]
    \centering
    \includegraphics[width=1\linewidth]{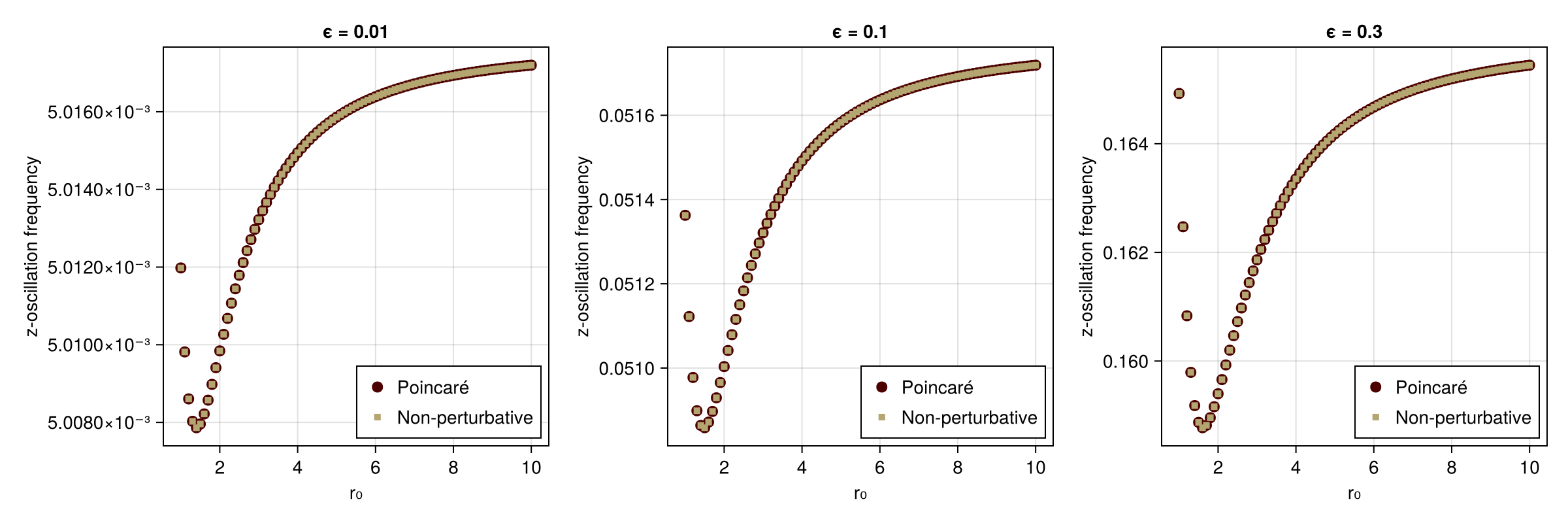}
    \caption{The circular points represent the frequency at which the full orbit $z$ intersects the Poincaré section. The rectangular points correspond to the frequency of the $z$ orbit in the non-perturbative model, initialized at $(\theta_0, z_0, p_{\perp 0}, p_{z0}) = (1, 1, 1.5, 0.5)$.}
    \label{fig:frequencycomparision}
\end{figure}



\section{Discussion}
Nonperturbative guiding center theory\cite{j_w_burby_nonperturbative_2025} offers a promising modeling alternative to traditional guiding center theory when $\epsilon$, the ratio of a particle's gyroradius to the scale length of the magnetic field, is only marginally small. While the model allows for nonperturbative $\epsilon$, it assumes existence of a hidden symmetry in the single-particle phase space that extends the perturbative hidden symmetry first identified by Kruskal.  This work rigorously justifies this assumption in idealized symmetric magnetic fields, including a slab configuration $\bm{B} = (1 + y)\bm{e}_z$ and the screw pinch. In particular, we have shown for each of these fields that there is an exact constant of motion for general $\epsilon$ that is agrees with Kruskal's adiabatic invariant series \emph{to all orders} when $\epsilon \ll 1$. Aside from the trivial case of a uniform magnetic field, we know of no other rigorous all-orders non-perturbative extensions of Kruskal's series expansion for the first adiabatic invariant. 

Qin and Davidson\cite{qin_exact_2006} previously found a family of exact constants of motion for a charged particle moving in a field of the form $\bm{B} = B(t)\,\bm{e}_z$, parameterized by solutions of an auxiliary ordinary differential equation. They showed that if $B(t) = \mathcal{B}(\epsilon t)$ is slowly varying and a slow solution of the auxiliary equation exists then there is an exact invariant asymptotic to Kruskal's adiabatic invariant at leading order in $\epsilon$. They did not demonstrate higher-order agreement with Kruskal's series, nor did they justify their assumption that a slow solution of the auxiliary equation exists. It would be interesting to determine if their assumption on existence of a slow solution is justified, and if agreement with Kruskal's series continues to higher order in perturbation theory. This problem is qualitatively different than the ones considered in this work because there is no obvious symmetry that implies integrability in the sense of Liouville.

The nonperturbative guiding center model was originally developed to extend the traditional guiding center model\cite{Cary_2009} into regimes where the particle gyroradius encroaches on the equilibrium scale length. It is natural to ask whether it also extends the so-called gyrocenter model for particles moving in fields that include a small-amplitude fluctuation with perpendicular length scales comparable to the gyroradius, as commonly employed in gyrokinetic modeling\cite{Brizard_2007}. Potential benefits of a nonperturbative model in this scenario would include allowing for nonperturbative amplitude of the small-scale fluctuations and a nonperturbative flute parameter $k_\parallel / k_\perp$. Since the gyrocenter model is fundamentally based on the same hidden symmetry principle\cite{Kruskal_1962,burby_general_2020,burby_normal_2021,burby_nearly_2023} as the guiding center model, and application of the nonperturbative model only requires existence of a nonperturbative hidden symmetry\cite{j_w_burby_nonperturbative_2025}, there is compelling reason to suspect an affirmative answer. It would be interesting to carefully investigate this question in future work.

\noindent\emph{Acknowledgements--} This material is based on work supported by the U.S. Department of Energy, Office of Science, Office of Advanced Scientific Computing Research, as a part of the Mathematical Multifaceted Integrated Capability Centers program, under Award Number DE-SC0023164. It was also supported by U.S. Department of Energy grant \# DE-FG02-04ER54742.

\bibliographystyle{unsrt}
\bibliography{cumulative_bib_file.bib}



\end{document}